\documentclass[num-refs]{wiley-article}

\usepackage{siunitx}
\usepackage{arydshln}
\usepackage{easybmat}
\usepackage{multirow,bigdelim}
\usepackage{tikz}
\usepackage{algorithmicx}
\usepackage[ruled]{algorithm}
\usepackage{algpseudocode}
\usepackage{algpascal}
\usepackage{algc}
\DeclarePairedDelimiter{\ceil}{\lceil}{\rceil}
\usetikzlibrary{automata,arrows,positioning,calc}

\newcommand{\RNum}[1]{\uppercase\expandafter{\romannumeral #1\relax}}

\papertype{Original Article}
\paperfield{Journal Section}

\title{An event-triggered transmission scheduling strategy for remote state estimation in the presence of an eavesdropper}

\author[1\authfn{1}]{Jingyi Lu}
\author[1\authfn{1}]{Alex S. Leong}
\author[1\authfn{1}]{Daniel E. Quevedo}

\affil[1]{Department of Electrical Engineering (EIM-E), Paderborn University, Paderborn, Germany}

\corraddress{Jingyi Lu, Department of Electrical Engineering (EIM-E), Paderborn University, Paderborn, 33098, Germany}
\corremail{jingyi.lu@upb.de}

\begin{document}

\maketitle

\begin{abstract}
We consider a remote state estimation problem in the presence of an eavesdropper over packet dropping links. A smart sensor transmits its local estimates to a legitimate remote estimator, in the course of which an eavesdropper can randomly overhear the transmission. This problem has been well studied for unstable dynamical systems, but seldom for stable systems. In this paper, we target at stable and marginally stable systems and aim to design an event-triggered scheduling strategy by minimizing the expected error covariance at the remote estimator and keeping that at the eavesdropper above a user-specified lower bound. To this end, we model the evolution of the error covariance as an infinite recurrent Markov chain and develop a recurrence relation to describe the stationary distribution of the state at the eavesdropper. Monotonicity and convergence properties of the expected error covariance are further investigated and employed to solve the optimization problem. Numerical examples are provided to validate the theoretical results.

\end{abstract}

\section{Introduction}
Wireless sensor networks (WSN) have received
significant attention in recent decades, and have been extensively applied to various fields such as smart grid, transportation, and industrial control. The techniques in WSN substantially increase system agility and offer more opportunities in remote sensing and remote control. Whereas, due to the broadcast nature of wireless communication, information can be intercepted by potential eavesdroppers, making private information leakage and potential economic loss inevitable \cite{chong2019tutorial,zhang2016privacy}. The security issue herein necessitates the study of information security and confidential communication. 

Traditionally, cryptography-based tools \cite{ganesan2003analyzing,lagendijk2012encrypted,darup2017towards} are utilized to implement information security. However, this approach is disadvantageous in some practical situations due to the large computation power required for encryption and decryption.  
Differential privacy \cite{dwork2011differential,cortes2016differential} is another approach to guarantee confidentiality. In differential privacy algorithms, additional noises are added to the signal to be sent. Consequently, both the information received by the eavesdropper and the estimator is disrupted, degrading estimation accuracy. Recently, information theoretic methods in the physical layer have gained increasing attention \cite{liang2008secure,shiu2011physical}. By exploiting physical layer characteristics of the wireless channel, this approach can effectively ensure a fairly good estimation at the authorized user while keep the estimation error at the eavesdropper above a certain level. Specifically, in \cite{reboredo2013filter,guo2017estimation,guo2016distortion}, estimation of constants and i.i.d. sources in the presence of eavesdroppers were studied in the context of physical layer security ideas using techniques such as transmit filter design , power control and addition of artificial noise. 

For dynamical system state estimation, a stochastic transmission strategy which randomly withholds sensor information was proposed in \cite{tsiamis2017state}. Under the condition that the user's packet reception rate is larger than the eavesdropper's interception rate, the eavesdropper's expected estimation error grows unboundedly while the user's expected error is kept bounded. Similar results are achieved in \cite{leong2018transmission} by a deterministic threshold policy without the prerequisite on reception rate, namely the unboundedness of the expected error at the eavesdropper can be obtained for all eavesdropping probabilities strictly less than one. Further in \cite{leong2019information}, upper and lower bounds on the information revealed to the eavesdropper are derived based on mutual information. Instead of looking at the expected error covariance, \citeauthor{tsiamis2019state}\cite{tsiamis2019state} proposed to drive the minimum mean square error at the eavesdropper to be unbounded by designing a coding mechanism with system dynamics and process noises exploited. Note that the unboundedness of the eavesdropper's covariance is achievable for unstable systems, but not stable or marginally stable systems. Since most of the open loop unstable systems are regulated by a controller when applied in practice, stable systems are commonly seen, and it is essential to explore such strategies for stable systems. 

Previously, a coding scheme was designed in \cite{tsiamis2018state} at the sensor for stable systems with an invertible dynamic matrix. The sensor encoded the current state as a weighted difference of a reference state, and transmitted the encoded state to the remote estimator. This scheme can ensure optimal estimation at the legitimate remote estimator, whilst making the error covariance at the eavesdropper asymptotically converge to its upper bound. In this scheme, an agreement should be reached on packet reception between the estimator and the sensor. Therefore, a precise acknowledgment from the estimator is required. 
In addition to coding the states, transmission scheduling over finite horizon can also be applied to stable systems to preserve privacy. In \cite{leong2018transmission}, an optimization problem was formulated by minimizing a linear combination of the expected error covariance of the remote estimator and the negative of the expected error covariance at the eavesdropper. Such a problem can be solved by dynamic programming. Due to the curse of dimensionality, its computation cost exponentially increases with the estimation horizon. Still, this work shows that the optimal solution essentially exhibits a thresholding behavior in the estimation error covariance. This structural property is illuminating. 

Motivated by this, in the present manuscript we propose a threshold-type event-triggered scheduling strategy specifically for stable systems. In particular, we formulate an optimization problem which takes the transmission threshold as a decision variable. The expected estimation error covariance at the legitimate estimator is minimized in the objective function and the eavesdropper's covariance is kept above a user specified value as a hard constraint. The crux of the design is how to solve the optimization problem. A major difficulty rests in the inaccessibility of the eavesdropper's covariance. Specifically, since the transmission command solely depends on the legitimate estimator's covariance and the communication channels are uncorrelated, the time at which the transmission occurs seems to be random to the eavesdropper, making it difficult to quantify the expected error covariance at the eavesdropper. To solve this problem, we combine the estimator's and eavesdropper's covariances as an artificial state whose evolution is modeled as an infinite Markov chain. Thereby, we set up the connection between the time instances at which transmission occurs and the estimation error covariance at the eavesdropper. Furthermore, it is shown that the Markov chain provides a recursive expression of the steady state distribution at the eavesdropper. Such a recursive expression can be viewed as a dynamic linear time invariant system with constant time delay, which opens the door to investigate monotonicity and convergence of the estimation error covariances. 
These properties allow us to solve the optimization problem by a simple bisection algorithm. 

The paper is organized as follows. Section 2 describes the system model and the problem of interest. Section 3 details the Markov chain modeling and analyzes properties on the eavesdropper's covariance. Section 4 provides a bisection algorithm to the optimization problem, as well as optimality analysis. Section 5 validates the theoretic results via numerical simulations. Finally, conclusions are drawn in Section 6. 

\textit{Notations}: 
For a matrix $X$, we use $X^\top$, $\text{tr}X$ and $\rho(X)$ to denote its transpose, trace and spectral radius. If $X$ is positive semi-definite, it is written as $X\geq0$. Denote $\mathbf{I}$ as the identity matrix. $\mathbb{N}$ is the set of non-negative integers. $\mathbb{E}$ is the expectation operator. $\mathbb{P}$ denotes the probability.

\section{Problem formulation}
\subsection{System model}
Consider a discrete time linear time invariant system as
\begin{align}
x_{k+1} = A x_{k} + w_k \nonumber	
\end{align}
where $x_k\in\mathbb{R}^{n_s}$ and $w_k$ is the process noise. The state is measured by a sensor through a noisy channel
\begin{align}
y_k = C x_k + v_k,\nonumber
\end{align}
with $y_k\in\mathbb{R}^{n_y}$ and $v_k$ denoting the measurement noise. We consider that the system is stable or marginally stable, namely the spectral radius of the matrix $A$, denoted as $\rho(A)$, is smaller than or equal to $1$. Moreover, it is postulated that \textbf{1)} $w_k$ and $v_k$ are mutually independent i.i.d Gaussian random processes with zero mean and covariance matrices $Q>0$ and $R>0$; \textbf{2)} the pair $(A,C)$ is detectable and the pair $(A,Q^{\frac{1}{2}})$ is stabilizable. 

\begin{figure}[t]
	\centering
	\includegraphics[width=0.5\textwidth]{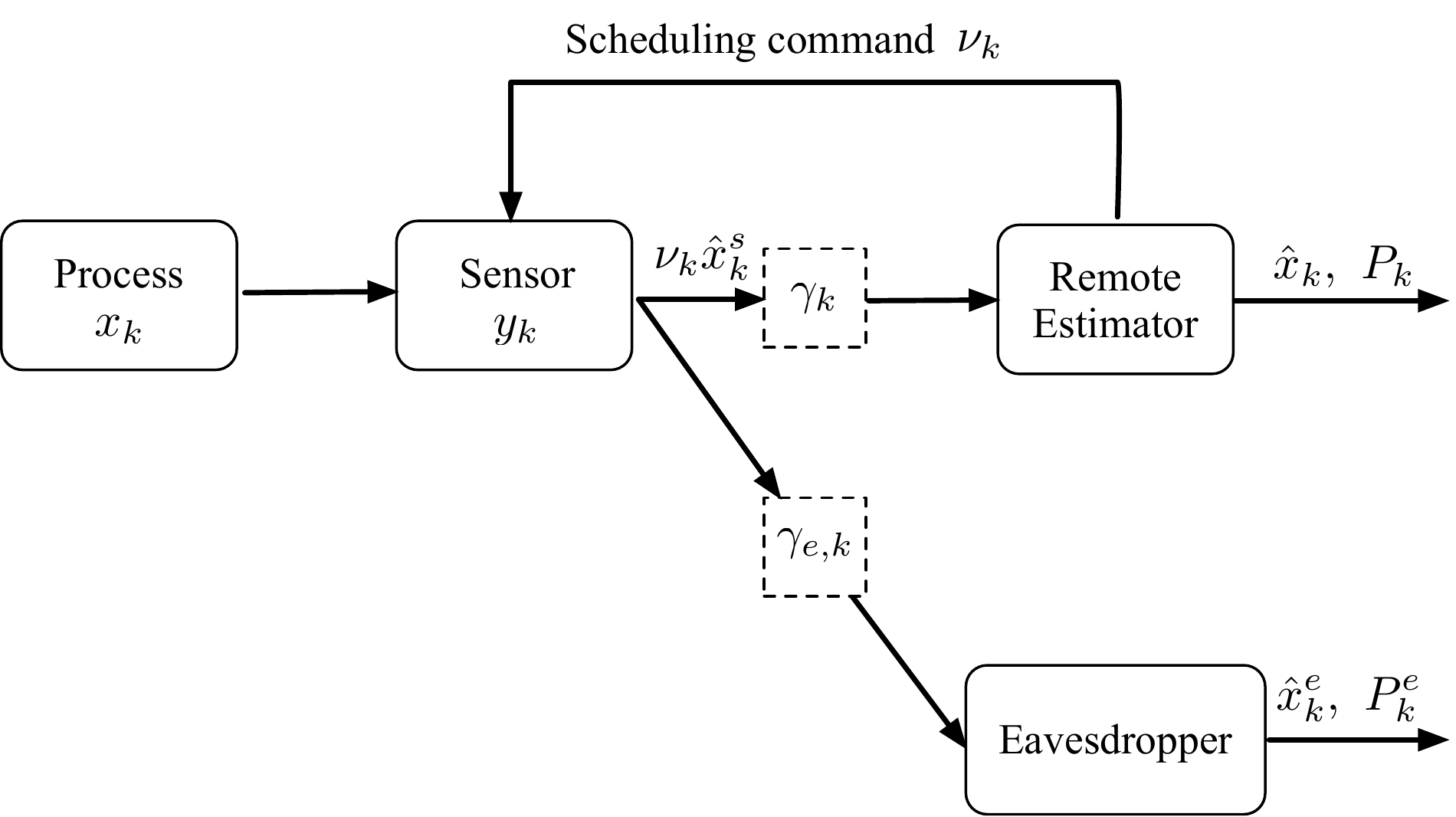}
	\caption{Remote state estimation with eavesdroppers} \label{system_model}
\end{figure}

As shown in Fig. \ref{system_model}, a sensor takes measurements of the process and sends its local state estimates to a legitimate remote estimator through a link with
 random packet dropouts. The transmission is controlled by the scheduling command $\nu_k$. When $\nu_k=1$, the local state estimate is transmitted from the sensor to the remote estimator. Each time the sensor transmits, an eavesdropper can successfully overhear the transmission with a certain probability via another packet dropping channel. Assume the sensor is a smart sensor capable of running a standard Kalman filter given in (\ref{kf1}) and (\ref{kf2}) and transmits the posterior local state estimate $\hat{x}_{k|k}^s$ to the receivers. 

\textbf{Prediction step}:
\begin{equation}
\begin{aligned}
\hat{x}_{k|k-1}^s \triangleq & \mathbb{E}[x_k|y_0,\dots,y_{k-1}]=A \hat{x}_{k-1|k-1}^s \\
\tilde{x}_{k|k-1}^s \triangleq & x_k- \hat{x}_{k|k-1}^s = A \tilde{x}^s_{k-1|k-1} + w_{k-1}\\
P_{k|k-1}^s \triangleq & \mathbb{E}[\tilde{x}_{k|k-1}^s\tilde{x}_{k|k-1}^{s^\top}|y_0,\dots,y_{k-1}] = A P_{k-1|k-1}^sA^\top +Q
\end{aligned} \label{kf1}
\end{equation}

\textbf{Update step}:
\begin{equation}
\begin{aligned}
\hat{x}_{k|k}^s \triangleq & \mathbb{E}[x_k|y_0,\dots,y_{k-1},y_k]= \hat{x}_{k|k-1}^s+ K_k(y_k-C\hat{x}_{k|k-1}^s) \\
\tilde{x}_{k|k}^s \triangleq & x_k- \hat{x}_{k|k}^s = (\mathbf{I}_{n_x}-K_kC)A\tilde{x}_{k-1|k-1}^s+(\mathbf{I}_{n_x}-K_kC)w_{k-1}-K_kv_k\\
P_{k|k}^s \triangleq & \mathbb{E}[\tilde{x}_{k|k}^s\tilde{x}_{k|k}^{s^\top}|y_0,\dots,y_{k-1}] = (\mathbf{I}_{n_x}-K_kC)P_{k|k-1}^s
\end{aligned} \label{kf2}
\end{equation}
with $K_k$ defined as
\begin{align}
K_k \triangleq P_{k|k-1}^sC^\top (CP_{k|k-1}^sC^\top +R)^{-1}.\nonumber
\end{align}

Let $\gamma_k$ be an indicator variable such that $\gamma_k=1$ if the transmission is successfully received by the remote estimator and $\gamma_k=0$ otherwise. Similarly, define $\gamma_k^e$ as the indicator variable at the eavesdropper, such that $\gamma_k^e=1$ if the sensor transmission at time $k$ is overheard by the eavesdropper and $\gamma_k^e=0$ otherwise. Assume ${\gamma_k}$ and $\gamma_k^e$ are mutually independent and both follow i.i.d. Bernoulli distribution with 
\begin{equation}
\begin{aligned}
&\mathbb{P}[\gamma_k=1~|~\nu_k=1]=\lambda \quad  \mathbb{P}[\gamma_k^e=1~|~\nu_k=1] =\lambda_e,  \\
&\mathbb{P}[\gamma_k=1~|~\nu_k=0]=0 \quad  \mathbb{P}[\gamma_k^e=1~|~\nu_k=0] =0.
\end{aligned} \label{dropout}
\end{equation}
Since the local Kalman filter converges to the steady state at an exponential rate, we assume that it runs at steady state for simplicity in presentation, namely, $P_{k|k}^s=\bar{P}$ with $\bar{P}$ obtained by solving a standard Riccati equation. For brevity, we denote the posterior estimate $\hat{x}_{k|k}^s$ as $\hat{x}_k^s$. As illustrated in \cite{leong2018transmission}, given that the scheduling command $\nu_k$ is independent of the state $x_k$, the optimal remote estimator has the form
\begin{align}
\hat{x}_{k}  =  \left\{\begin{array}{cl} \hat{x}_{ k}^s, & 
\gamma_{k} \nu_k= 1 \\ A \hat{x}_{k-1}, & 	\gamma_{k} \nu_k= 0, 
\end{array}  \right. \quad 
P_{k}  = \left\{\begin{array}{cl}  \overline{P}  &   \textnormal{if
	$\gamma_{k} \nu_k=  1$} \\ 
f(P_{k-1}),  &
\textnormal{otherwise},   \end{array} \right. ,\label{estimator}
\end{align}
where $	f(X) \triangleq A X A^T + Q$. Similarly, the optimal estimator at the eavesdropper has the form
\begin{align}
\hat{x}_{k}^e  =  \left\{\begin{array}{cl} \hat{x}_{ k}^s, & 
\gamma_{k}^e \nu_k= 1 \\ A \hat{x}_{k-1}^e, & 	\gamma_{k}^e \nu_k= 0, 
\end{array}  \right. \quad 
P_{k}^e  = \left\{\begin{array}{cl}  \overline{P}  &   \textnormal{if
	$\gamma_{k} \nu_k=  1$} \\ 
f(P_{k-1}^e),  &
\textnormal{otherwise},   \end{array} \right. .\nonumber
\end{align}
Define a set
\begin{align}
\mathbb{S} \triangleq \{ \overline{P}, f(\overline{P}), f^2 (\overline{P}),f^3 (\overline{P}), \dots\}.\nonumber
\end{align}
Here $f^n(\cdot)$ denotes the $n-$ fold composition of $f(\cdot)$ with $f^0(X)=X$. According to \cite{shi2012scheduling}, $\mathbb{S}$ is totally ordered with
\begin{align}
\bar{P}\leq f(\bar{P})\leq f^2(\bar{P})\leq \dots.\label{p_mono}
\end{align}
It can be verified that the set $\mathbb{S}$ consists of all possible values of $P_k$ and $P_k^e$. Denote $n_k$ and $n_k^e$ as the number of time steps since the last successful receipt by the remote estimator and the eavesdropper respectively, namely 
\begin{align}
n_k  \triangleq \min \{\tau \geq 0: \gamma_{k-\tau}=1\} \quad n_k^e \triangleq \min \{\tau \geq 0: \gamma_{k-\tau}^e=1\}.\nonumber
\end{align} 
Then, $P_k$ and $P_k^e$ can be expressed in terms of $n_k$ and $n_k^e$ as
\begin{align}
P_k = f^{n_k}(\bar{P}) \quad P_k^e = f^{n_k^e}(\bar{P}).\nonumber
\end{align}
\begin{remark}
As shown in Fig. \ref{system_model}, it is assumed that the channel used to transmit the scheduling command is reliable with no packet dropout. This assumption can be easily relaxed by equivalently attributing the packet dropouts in the command channel to the state transmission channel. Specifically, let $\gamma_k^v$ be an indicator variable such that $\gamma_k^v=1$ if the scheduling command is successfully transmitted and $\gamma_k^v=0$ otherwise. Here $\gamma_k^v$ should be independent of $\gamma_k$ and $\gamma_k^e$. Then, the optimal remote estimator is cast as
\begin{align}
\hat{x}_{k}  =  \left\{\begin{array}{cl} \hat{x}_{ k}^s, & 
\gamma_{k}\gamma_k^v \nu_k= 1 \\ A \hat{x}_{k-1}, & 	\gamma_{k} \gamma_k^v\nu_k= 0. 
\end{array}  \right. \quad 
P_{k}  = \left\{\begin{array}{cl}  \overline{P}  &   \textnormal{if
	$\gamma_{k} \gamma_k^v\nu_k=  1$} \\ 
f(P_{k-1}),  &
\textnormal{otherwise},   \end{array} \right. .\nonumber
\end{align}
If we introduce a new binary variable $\gamma_k^s = \gamma_k\gamma_k^v$, it can be noticed that the estimator has the same form as (\ref{estimator}). Therefore, packet dropouts in the command channel can also be addressed in the current framework. 
\end{remark}
\nopagebreak[4]
\subsection{Problem of interest}
In this paper, we aim to design a threshold type event-triggerred transmission policy as 
\begin{align}
\nu_k(P_{k-1})  =  \left\{\begin{array}{cl} 1, & 
P_{k-1} \geq f^{\bar{t}}(\bar{P}) \\ 0, & P_{k-1} < f^{\bar{t}}(\bar{P}), 
\end{array}  \right. \label{policy}
\end{align}
which implies that the sensor transmits if and only if the expected error covariance at the legitimate estimator is greater than $f^{\bar{t}}(\bar{P})$. We endeavor to seek the optimal threshold $\bar{t}$ such that the average of the expected error covariance is minimized over an infinite horizon and that at the eavesdropper is kept larger than a user specified value. In particular, such a design can be formulated as an optimization problem with the transmission threshold $\bar{t}$ taken as a decision variable:

\textbf{\underline{Problem 1:}}
\begin{align}
\min_{\bar{t}} ~& ~J(\bar{t})=\lim_{K\rightarrow\infty}\sum_{k=0}^K\frac{1}{K}\text{tr}\mathbb{E}[P_{k}] \nonumber\\
\text{s.t.} ~~& ~\lim_{K\rightarrow\infty}\sum_{k=0}^K\frac{1}{K}\text{tr}\mathbb{E}[P_{k}^e ] \geq \underline{b},\nonumber
\end{align}
Here $\underline{b}$ is a user specified lower bound on the error covariance at the eavesdropper.
\begin{remark}
If $\rho(A)<1$, according to \cite{shi2012scheduling}, the averaged expected error covariance is bounded. Therefore, Problem 1 cannot be guaranteed to be feasible for all values of $\underline{b}$. A necessary and sufficient condition on $\underline{b}$ to ensure feasibility is that 
\begin{align}
\underline{b} \leq \text{tr} f^{\infty}(\bar{P})= \sum_{k=0}^{\infty} \text{tr} A^k\bar{P} \left(A^\top\right)^k,\nonumber
\end{align}
If $\rho(A)=1$, there is no such restriction since $\text{tr} f^k(\bar{P})$ grows to infinity when $k$ goes to infinity. 
\end{remark}
 In \cite{leong2015sensor}, it is shown that the evolution of the error covariance at the remote estimator can be modeled as a recurrent Markov chain whose stationary probability follows
\begin{align}
\pi_j(\bar{t}) = \left\{\begin{array}{ccl}
\frac{\lambda}{\lambda \bar{t}+1} &,& j = 0,\dots,\bar{t}\\
\frac{(1-\lambda)^{j-\bar{t}}\lambda}{\lambda \bar{t}+1} &,& j=\bar{t}+1,\bar{t}+2,\dots 
\end{array}\right. \label{dis_pi}
\end{align}
Correspondingly, the objective function $J(\bar{t})$ can be computed by
\begin{align}
J(\bar{t}) = \sum_{j=0}^{\bar{t}}\frac{\lambda}{\lambda \bar{t}+1} \text{tr}(f^j(\bar{P})) +  \sum_{j=\bar{t}+1}^{\infty} \frac{(1-\lambda)^{j-\bar{t}}\lambda}{\lambda \bar{t}+1} \text{tr}(f^j(\bar{P})).
\end{align}
Also, by closely examining the distribution in (\ref{dis_pi}), we conclude that
\begin{proposition}
	\label{prop1}
	The function $J(\bar{t})$ monotonically increases with the threshold $\bar{t}$, i.e.
	\begin{align}
	J(\bar{t}_1) \leq J(\bar{t}_2) \nonumber
	\end{align}
	for any given integers $0\leq\bar{t}_1\leq\bar{t}_2$. \label{prop_mono}
\end{proposition}
A brief proof is provided in Appendix \ref{pf_prop1}. 

In view of the monotonicity of $J(\bar{t})$, minimizing $J(\bar{t})$ is equivalent to minimizing $\bar{t}$. Therefore, Problem 1 can be simplified as

\textbf{\underline{Problem 2:}}
\begin{align}
\min_{\bar{t}} ~& ~\bar{t}\nonumber\\
\text{s.t.} ~~& ~\lim_{K\rightarrow\infty}\sum_{k=0}^K\frac{1}{K}\text{tr}\mathbb{E}[P_{k}^e ] \geq \underline{b},\label{infinite}
\end{align}
Whilst in \cite{leong2015sensor}, we developed an expression for the stationary estimation error covariance at the legitimate estimators, expressions for $P_k^e$ were not obtained. To be able to enforce the security constraint  in (\ref{infinite}), we will next show: 
\begin{itemize}
	\item[1)] how to evaluate the expected error covariance at the eavesdropper;
	\item[2)] how to efficiently solve Problem 2.
\end{itemize}

\section{Distribution of expected error covariance at the eavesdropper}
In this section, we employ a Markov chain to model the state evolution when the transmission is scheduled by the threshold policy in (\ref{policy}). Then, we show that the stationary distribution of the expected error covariance at the eavesdropper can be calculated in a recursive manner, based on which, properties on convergence and monotonicity are further provided. 

As aforementioned, the key challenge lies in the tricky correlation between the eavesdropper's error covariance and the time instances at which a transmission occurs. To tackle this problem, we define an artificial state 
\begin{align}s_k=(i,j), \nonumber
\end{align}
associating with the state that $P_k=f^{i}(\bar{P})$ and $P_k^e=f^j(\bar{P})$. Here $i,j\in\mathbb{N}$, with $\mathbb{N}$ denoting the set of non-negative integers.
 Correspondingly, the threshold policy in (\ref{policy}) can be rewritten as
\begin{align}
\nu_k(i)  =  \left\{\begin{array}{cl} 1, & 
i \geq \bar{t} \\ 0, & i< \bar{t}
\end{array}  \right.. \label{policy2}
\end{align}
According to (\ref{policy2}), and taking into account packet dropouts given in (\ref{dropout}), the state evolution must follow the equations in (\ref{transition1}) and (\ref{transition2}),
\begin{align}
&\begin{array}{ll}
 i \geq \bar{t}: & \mathbb{P}(s_{k+1}=(0,0)\mid s_k=(i,j)) = \mathbb{P}(\gamma_k=1,\gamma_k^e=1) =\lambda\lambda_e\\
 & \mathbb{P}(s_{k+1}={(i+1,0)}\mid s_k={(i,j)}) = \mathbb{P}(\gamma_k=0,\gamma_k^e=1) = (1-\lambda)\lambda_e\\
 & \mathbb{P}(s_{k+1}={(0,j+1)}\mid s_k={(i,j)}) =\mathbb{P}(\gamma_k=1,\gamma_k^e=0)= \lambda(1-\lambda_e) \\
 & \mathbb{P}(s_{k+1}={(i+1,j+1)}\mid s_k={(i,j)}) = \mathbb{P}(\gamma_k=0,\gamma_k^e=0)= (1-\lambda)(1-\lambda_e)
\end{array} \label{transition1} \\
& \begin{array}{cc}
i<\bar{t}: & \mathbb{P}(s_{k+1}={(i+1,j+1)}\mid s_k={(i,j)}) = 1
\end{array} \label{transition2}
\end{align} 
based upon which, a time-homogeneous recurrent Markov chain is established in Lemma \ref{markov_chain} to model  the state evolution. 
\begin{lemma}
	\label{markov_chain}
	Define an integer set $\Omega_i$ as
	$\Omega_i = [0,i-\bar{t}-1]\cup [i,\infty)$.
	The states $(i,j)$ with $i\in[0,\infty)$ and $j\in\Omega_i$ compose an irreducible, aperiodic and recurrent Markov chain,
	whose probability transition matrix is given as
	\begin{align}
	\mathbf{P}=\left(\begin{array}{cccccccc}
	\mathbf{0} & \mathbf{I} & 	\mathbf{0} &	\mathbf{0} & \dots & \dots & 	\mathbf{0} & 	\mathbf{0}\\
	\mathbf{0} & \mathbf{0} & 	\ddots &	\mathbf{0} & \dots & \dots & 	\mathbf{0} & 	\mathbf{0}\\
	\mathbf{0} & \mathbf{0} & 	\dots &	\mathbf{I} & \mathbf{0} & \dots & 	\mathbf{0} & 	\mathbf{0}\\
	\mathbf{T}_{\bar{t},0} & \mathbf{0} & 	\dots &	\mathbf{0} & \mathbf{T}_c &  	\mathbf{0} & \dots  & 	\mathbf{0}\\
	\mathbf{T}_{\bar{t}+1,0} & \mathbf{0} & 	\dots &	\mathbf{0} &  	\mathbf{0} &\mathbf{T}_c & \dots  & 	\mathbf{0}\\
	\vdots & \vdots & 	\vdots & \vdots & 	\vdots & \vdots & 	\vdots & \vdots 
	\end{array}\right).\nonumber
	\end{align}
	Here matrices $\mathbf{T}_{\bar{t}+i,0}$ and $\mathbf{T}_c$ are defined as
	\begin{align}
	& \mathbf{T}_{\bar{t}+i,0}=  \left(
	\begin{array}{rccc}
	\begin{array}{r}
	\\[-33mm]\ldelim\{ {4}{-2mm}[$i$] \\  \end{array}  & \begin{array}{c}	\begin{array}{ccccc}
	\lambda\lambda_e & 	\lambda(1-\lambda_e) & 0  & \dots & 0\\
	\lambda\lambda_e &  0 &	\lambda(1-\lambda_e)  & \dots & 0\\
	\vdots & \vdots & \vdots & \ddots & \vdots\\
	\lambda\lambda_e &  0 & \dots & \dots & \lambda(1-\lambda_e)
	\\ \hdashline[2pt/2pt]
	\lambda\lambda_e & 	0 & 0  & \dots & 0\\
	\lambda\lambda_e &  0 &	0  & \dots & 0\\
	\vdots & \vdots & \vdots & \ddots & \vdots
	\end{array} \end{array}  & \overbrace{\text{\huge 0}}^{\bar{t}} & \begin{array}{c}\\\text{\huge 0} \\ [10.5mm]\begin{array}{ccc}
	\hdashline[2pt/2pt] \lambda(1-\lambda_e) & 0 & \dots \\
	0 &  \lambda(1-\lambda_e) & \dots\\
	\vdots & \vdots & \ddots
	\end{array} \end{array}
	\end{array}
	\right) \label{pt1}
	\end{align}
	and 
	\begin{align}
	\mathbf{T}_c = \left(	\begin{array}{ccccc}
	(1-\lambda)\lambda_e & 	(1-\lambda)(1-\lambda_e) & 0  & \dots & 0\\
	(1-\lambda)\lambda_e &  0 &	(1-\lambda)(1-\lambda_e)  & \dots & 0\\
	\vdots & \vdots & \vdots & \ddots & \vdots
	\end{array} \right). \label{pt2}
	\end{align}
\end{lemma}
\begin{proof}	
	First, it can be checked from (\ref{transition1}) and (\ref{transition2}) that states $(i,j)$ with $j\in[\max{(i-\bar{t},0)},i]$ are not recurrent, since the state $(0,0)$ is obviously recurrent and for a state $s_0=(0,0)$, the $k$-step transition probability $\mathbb{P}(s_k=(i,j),j\in[i-\bar{t},i-1]\mid s_0=(0,0))=0$ for any positive integer $k$. Define $\mathcal{S}_k=(i,\cdot)$ as a collection of states $(i,j)$ with $j\in\Omega_i$, ruling out all the non-recurrent states. Then, we have
	\begin{align}
	(i,\cdot)=\left\{\begin{array}{lc}
	\{(i,i),~(i,i+1),~(i,i+2),\dots\} & i<\bar{t}+1. \\
	\{(i,0),~(i,1),~(i,i-\bar{t}-1),(i,i),(i,i+1),\dots\} & i\geq \bar{t}+1.
	\end{array}\right.\label{range_j}
	\end{align}
Given that $\mathcal{S}_k=(i,\cdot)$, denote $s_{k}=(i,\cdot)_j$ as the $j$-th element in $\mathcal{S}_{k}$. As shown in (\ref{range_j}), due to the discontinuity of $j$, the state $(i,\cdot)_j$ is generally different from the state $(i,j)$. Following Eq. (\ref{transition1}) and Eq. (\ref{transition2}), it can be derived that 
if $i<\bar{t}$ \begin{align}
	\mathbb{P}(s_{k+1}=(i+1,\cdot)_j\mid s_{k}=(i,\cdot)_j)=1,\label{tran1}
	\end{align} and if $i\geq\bar{t}$,
	\begin{equation}
	\left\{\begin{array}{ll}
	\mathbb{P}(s_{k+1}=(i+1,\cdot)_j\mid s_k=(i,\cdot)_j)= \mathbb{P}(\gamma_k=0,\gamma_k^e=0)= (1-\lambda)(1-\lambda_e)  &  \\
	\mathbb{P}(s_{k+1}=(i+1,\cdot)_0\mid s_k=(i,\cdot)_j)= \mathbb{P}(\gamma_k=0,\gamma_k^e=1)= (1-\lambda)\lambda_e  &  \\
	\mathbb{P}(s_{k+1}=(i,\cdot)_0\mid s_k=(i,\cdot)_j)= \mathbb{P}(\gamma_k=1,\gamma_k^e=1)= \lambda\lambda_e  & \\
	\mathbb{P}(s_{k+1}=(0,\cdot)_{j+1}\mid s_k=(i,\cdot)_j)= \mathbb{P}(\gamma_k=1,\gamma_k^e=0) =\lambda(1-\lambda_e)& j< i\\
	\mathbb{P}\left(s_{k+1}=(0,\cdot)_{j+\bar{t}+1}\mid s_k=(i,\cdot)_j\right)= \mathbb{P}\left(\gamma_k=1,\gamma_k^e=0\right) =\lambda(1-\lambda_e)& j\geq i
	\end{array} \right. . \label{tran2}
	\end{equation}
According to (\ref{tran2}), a Markov chain revealing the state evolution is established in Fig. 2 with transition probability matrix defined in (\ref{pt1}) and (\ref{pt2}). $\hfill\square$
	
	\begin{figure}[t]
		\centering
		\label{markov}
		\begin{tikzpicture}[->, >=stealth', auto, semithick, node distance=1.7cm]
		\tikzstyle{every state}=[fill=white,draw=black,thick,text=black,scale=1]
		\node[state,minimum size=1.1cm]    (A)                     {\footnotesize $(0,\cdot)$};
		\node[state,minimum size=1.1cm]    (B)[right of=A]   {\footnotesize$(1,\cdot)$};
		\node[state,minimum size=1.1cm]    (C)[right of=B]   {\footnotesize$(2,\cdot)$};
		\node[draw=none,right=of C]           (C-D) {$\dots$};
		\node[state,minimum size=1.1cm]    (D)[right of={C-D}]   {\footnotesize$(\bar{t},\cdot)$};
		\node[state,minimum size=1.1cm]    (E)[right of=D]   {\footnotesize $(\bar{t}+1,\cdot)$};
		\node[draw=none,right=of E]           (F) {$\dots$};
		\path
		(A) edge   node{$\mathbf{I}$}      (B)
		(B) edge                node{$\mathbf{I}$}           (C)
		(C) edge                node{$\mathbf{I}$}           ({C-D})
		({C-D}) edge    node{$\mathbf{I}$}     (D)
		(D) edge node{$\mathbf{T}_c$} (E)
		edge[bend left,above]  node{$\mathbf{T}_{\bar{t},0}$} (A)
		(E) edge node{$\mathbf{T}_c$} (F)
		edge[bend left]  node{$\mathbf{T}_{\bar{t}+1,0}$} (A);
		\end{tikzpicture}
		\caption{Markov chain}
	\end{figure}
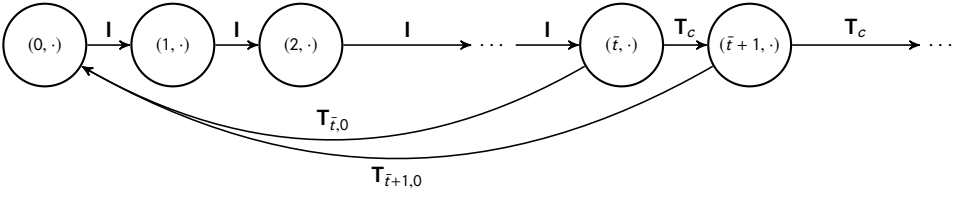
\end{proof}
Next, we will endeavor to derive the stationary distribution of the states in this Markov chain to obtain an explicit expression of the average expected error covariance. In particular, denote the stationary probability corresponding to the state $(i,j)$ as $\phi_{i,j}(\bar{t})$, which intrinsically represents the long run proportion of time that the chain stays at the state $(i,j)$, i.e.
\begin{align}
\phi_{i,j}(\bar{t})=\lim_{m\rightarrow\infty}\frac{\sum_{k=0}^{m} \mathcal{I}\left(s_k=(i,j)\right)}{m} \quad\text{w.p.1}\label{long}
\end{align} 
with $\mathcal{I}$ defined as an indicator function. Thereupon, the average of the expected error covariance at the eavesdropper can be computed as
\begin{align}
\lim_{K\rightarrow\infty}\frac{1}{K}\sum_{k=1}^K \text{tr} \mathbb{E}[P_k^e]=\sum_{j=0}^{\infty}\sum_{i=0}^{\infty} \phi_{i,j}(\bar{t}) \text{tr} (f^j(\bar{P})). \label{average}
\end{align} 
Note from (\ref{long}) that transient states should have zero stationary probability. 

Define a set $\Gamma_j$ as 
\begin{align}\Gamma_j=[0,j]\cup [\bar{t}+j+1,\infty).\nonumber
\end{align}
It can be checked from the Markov chain that for a given $j$, if the state $(i,j)$ is recurrent, then  $i\in\Gamma_j$. If we define 
\begin{align}
\omega_j(\bar{t}) = \sum_{i\in\Gamma_j} \phi_{i,j}(\bar{t}),\nonumber
\end{align}
the average of expected error covariance in (\ref{average}) can be simplified as
\begin{align}
\lim_{K\rightarrow\infty}\frac{1}{K}\sum_{k=1}^K \text{tr} \mathbb{E}[P_k^e]=\sum_{j=0}^{\infty}\omega_j(\bar{t}) \text{tr} (f^j(\bar{P})). \label{infinite_sum}
\end{align} 
If $\omega_j(\bar{t})$ is known for a given $\bar{t}$, the error covariance at the eavesdropper can be computed. In the following, we propose a recursive approach to calculate $\omega_j$. For brevity, $\bar{t}$ is omitted whenever no ambiguity arises. 

First, we look at the computation of $\phi_{i,j}$. Define a vector with infinite dimension as
\begin{align}
\Phi=\left[\begin{array}{cccc}
\Phi_0 & \Phi_1 & \Phi_2 & \dots
\end{array}\right].\nonumber
\end{align}
Herein, $\Phi_i$ is defined as
\begin{align}
\Phi_i =\left[\begin{array}{ccccccc}
\phi_{i,0} & \phi_{i,1} & \dots & \phi_{i,i-\bar{t}-1} & \phi_{i,i} & \phi_{i,i+1} & \dots
\end{array}\right]
\end{align}
which collects all $\phi_{i,j}$ with $j\in\Omega_i$. Essentially, the vector $\Phi$ is an assembly of the stationary probabilities associating with recurrent states. Thus, it satisfies that 
\begin{align}
\Phi = & \Phi \mathbf{P} \label{mc1}.
\end{align}	
As per the way that the state $(i,j)$ is defined, we have
 \begin{align}
\sum_{j\in\Omega_i} \phi_{ij} = &\pi_i, ~\forall i\in\mathbb{N},\label{mc3}\\
\sum_{i,j\in\Omega_i} \phi_{ij}  = & 1.\label{mc2}
\end{align}
By combining Eq. (\ref{mc1})-Eq. (\ref{mc3}), we derive a recurrence relation of $\phi_{0,j}$ for all $j\in\Omega_0$, which is given in Lemma \ref{lemma2}. 
\begin{lemma}
	\label{lemma2}
	Define $\alpha = (1-\lambda)(1-\lambda_e)$ and $\beta = \lambda(1-\lambda_e)$. $\phi_{0,j}$ satisfies that
	\begin{equation}
	\begin{aligned}
	\phi_{0,0} = & \frac{\lambda_e\lambda}{\lambda \bar{t} + 1} \\
	\phi_{0,j+1} = & \alpha \phi_{0,j} \quad j< \bar{t}  \\
	\phi_{0,j+1} = & \alpha \phi_{0,j} +\beta \phi_{0,j-\bar{t}} \quad j\geq \bar{t} 
	\end{aligned}  \label{lemma23}
	\end{equation}
\end{lemma}
A detailed proof is provided in Appendix \ref{pfLemma2}. Note that if we regard $\phi_{0,j}$ as a state of a dynamical system, the recursion in (\ref{lemma23}) constitutes a scalar linear time invariant system with a single constant time delay. Given that $\alpha+\beta<1$, we show in Proposition \ref{prop3} that $\phi_{0,j}$ exponentially converges to $0$ by employing Lyapunov-based methods for linear systems. 

\begin{proposition}
Define $\gamma = (1-\lambda_e)^{\frac{1}{2(\bar{t}+1)}}$. Then
	\begin{itemize}
		\item [(\RNum{1})] for each integer $j> \bar{t}$, $\phi_{0,j}$ satisfies that  \begin{align}\phi_{0,j}\leq \gamma^{j-\bar{t}}\phi_{0,0}; \label{phi_bounded}\end{align} 
		\item [(\RNum{2})] $\phi_{0,j}$ converges to $0$ exponentially as $j$ goes to infinity, i.e. $
		\lim_{j\rightarrow\infty} \phi_{0,j} = 0.
		$
		\item [(\RNum{3})] the sum $\sum_{j=0}^{n}\phi_{0,j}$ monotonically decreases with the threshold $\bar{t}$ for any $n\geq0$, i.e. 
		\begin{align}
		\sum_{j=0}^{n}\phi_{0,j}(\bar{t}_1) \geq \sum_{j=0}^{n}\phi_{0,j}(\bar{t}_2) \label{phi_mono}
		\end{align} 
		with $\bar{t}_1\leq \bar{t}_2$,
	\end{itemize}
	\label{prop3}
\end{proposition}
Please refer to Appendix \ref{pfprop3} for the proof. 

Based on the results in Lemma \ref{lemma2} and Proposition \ref{prop3}, a recurrence relation on $\omega_j$ can be obtained. Similar to $\phi_{0,j}$, $\omega_{j}$ is shown to converge to $0$ superlinearly and the average of the expected error covariance increases monotonically with the threshold $\bar{t}$. 
\begin{theorem}
	\label{thm5}
	The stationary distribution of the expected error covariance at the eavesdropper can be computed recursively as
	\begin{align}
	\begin{array}{ll}
	\omega_0 =  \frac{\lambda_e}{\lambda\bar{t}+1} & \\
	\omega_{j} =  \alpha \omega_{j-1}  + \phi_{0,0} & j\in[1,\bar{t}]
	\\
	\omega_{j}=\alpha\omega_{j-1}+\phi_{0,0}-\lambda_e\sum_{l=0}^{j-1-\bar{t}}\phi_{0,l} & j\in[\bar{t}+1,\infty).
	\end{array} \label{recursion}
	\end{align}
	Moreover, we have 
	\begin{itemize}
		\item[(\RNum{1})]
		$w_j$ monotonically decreases, i.e. \begin{align}
		w_j\leq w_{j-1} \label{thm_mono}
		\end{align}for any integer $j>0$. 		
		\item[(\RNum{2})]
		for any integer $j\geq N \geq 3\bar{t}+1$, $\omega_j$ is bounded by
		\begin{align}
		\omega_j\leq \alpha^{j-N} \omega_N + (j-N)(1+\beta \bar{t})\phi_{0,0}\gamma^{j-1-3\bar{t}}.\label{omega_bound}
		\end{align}	
			\item[(\RNum{3})] the cumulative distribution function monotonically decreases with the threshold $\bar{t}$, i.e.
		\begin{align}
		\sum_{j=0}^n \omega_j(\bar{t}_1) \geq \sum_{j=0}^n \omega_j(\bar{t}_2) \label{omega_mono}
		\end{align}
		for any $n\geq 0$ and $0\leq \bar{t}_1\leq \bar{t}_2$.
		\item[(\RNum{4})]
		Regard the average of expected error covariance $\lim_{K\rightarrow\infty}\frac{1}{K}\sum_{k=1}^K \text{tr} \mathbb{E}[P_k^e]$ as a function of $\bar{t}$. Then, it monotonically increases with the threshold $\bar{t}$, i.e.
		\begin{align}
		\lim_{K\rightarrow\infty}\frac{1}{K}\sum_{k=1}^K \text{tr} \mathbb{E}[P_k^e](\bar{t}_1) \leq \lim_{K\rightarrow\infty}\frac{1}{K}\sum_{k=1}^K \text{tr} \mathbb{E}[P_k^e](\bar{t}_2),\label{mono_average}
		\end{align}	
		with $0\leq \bar{t}_1\leq \bar{t}_2$.
	\end{itemize}
\end{theorem}
		Detailed proof is provided in Appendix \ref{pfThm5}. Note that since $0<\alpha<\gamma<1$, it can be easily seen from (\ref{omega_bound}) that
\begin{align}
\lim_{j\rightarrow \infty} \omega_j = 0.\label{omega_converge}
\end{align}

\section{Optimization Algorithm}
By virtue of the monotonicity of $\lim_{K\rightarrow\infty}\frac{1}{K}\sum_{k=1}^K \text{tr} \mathbb{E}[P_k^e](\bar{t})$ in (\ref{mono_average}), we can apply a bisection algorithm to solve the optimization problem. Before that, we truncate the infinite series in (\ref{infinite_sum}) and construct a lower bound of $\lim_{K\rightarrow\infty}\frac{1}{K}\sum_{k=1}^K \text{tr} \mathbb{E}[P_k^e](\bar{t})$ in Theorem \ref{thm6} to avoid the summation of infinite series.

\begin{theorem}
	\label{thm6}
	Define a function $L(N,\bar{t})=\sum_{j=0}^{N} \omega_j \text{tr} f^j(\bar{P})+\left(1-\sum_{j=0}^{N} \omega_j\right)\text{tr} f^{N+1}(\bar{P})$. Then, we have
	\begin{itemize}
		\item[(\RNum{1})] for any integer $N>0$, $L(N,\bar{t})$ is a lower bound of $\lim_{K\rightarrow\infty}\frac{1}{K}\sum_{k=1}^K \text{tr} \mathbb{E}[P_k^e](\bar{t})$, i.e.
		\begin{align}
		L(N,\bar{t}) \leq  \lim_{K\rightarrow\infty}\frac{1}{K}\sum_{k=1}^K \text{tr} \mathbb{E}[P_k^e](\bar{t}); \label{lower}
		\end{align}
		\item[(\RNum{2})] there exists a function $b(N,\bar{t})$, which converges to $0$ as $N\rightarrow\infty$, upper bounding the difference between  $\lim_{K\rightarrow\infty}\frac{1}{K}\sum_{k=1}^K \text{tr} \mathbb{E}[P_k^e](\bar{t})$ and $L(N,\bar{t})$, i.e. 
		\begin{align}
		\lim_{K\rightarrow\infty}\frac{1}{K}\sum_{k=1}^K \text{tr} \mathbb{E}[P_k^e](\bar{t})-L(N,\bar{t}) \leq b(N,\bar{t})
		\end{align}
		with \begin{align}\lim_{N\rightarrow \infty} b(N) = 0.\nonumber
		\end{align}
		\item[(\RNum{3})] the function $L(N,\bar{t})$ monotonically increases with $\bar{t}$, i.e.
		\begin{align}
		L(N,\bar{t}_1) \leq  L(N,\bar{t}_2),\label{s_mono}
		\end{align}
		with $\bar{t}_1\leq\bar{t}_2$.
	\end{itemize}  \label{lower_bound}
\end{theorem}
Please refer to Appendix \ref{pfThm6} for a detailed proof. Here $N$ is the truncation horizon. Theorem \ref{lower_bound} indicates that the lower bound $L(N,\bar{t})$ is a good approximation to $\lim_{K\rightarrow\infty}\frac{1}{K}\sum_{k=1}^K \text{tr} \mathbb{E}[P_k^e]$, since the approximation error diminishes as $N$ increases to infinity. Thereafter, we replace  $\lim_{K\rightarrow\infty}\frac{1}{K}\sum_{k=1}^K \text{tr} \mathbb{E}[P_k^e](\bar{t})$ in (\ref{infinite})
with $L(N,\bar{t})$ to simplify Problem 2 as

\textbf{\underline{Problem 3:}}
\begin{align}
\min_{\bar{t}} ~& ~\bar{t}\nonumber\\
\text{s.t.} ~~& ~L(N,\bar{t})\geq \underline{b}.\label{finite}
\end{align}
The monotonicity in (\ref{s_mono}) allows us to solve Problem 3 by a simple bisection algorithm  with a finite number of iterations. Details are given in Algorithm 1.

\alglanguage{pseudocode}
\begin{algorithm}[h]
	\caption{A bisection algorithm for Problem 3}
	\label{alg1}
	\begin{algorithmic}[1]
		\State \textbf{Input} $A$, $Q$, $\bar{P}$, $\lambda$, $\lambda_e$, $\underline{b}$
		\State \textbf{Initialize} integer $N$, $t_{\max}$, $t_{\min}$ such that $L(N,t_{\min})<\underline{b}$ and $L(N,t_{\max})>\underline{b}$.
		\While {($t_{\min}<t_{\max}-1$)}
		\If{$L(N,\ceil[\bigg]{\frac{t_{\min}+t_{\max}}{2}})<\underline{b}$}
		\State Set $t_{\min}=\ceil[\bigg]{\frac{t_{\min}+t_{\max}}{2}}$
		\Else 
		\State Set $t_{\max}=\ceil[\bigg]{\frac{t_{\min}+t_{\max}}{2}}$
		\EndIf
		\EndWhile
		\State Assign $\bar{t}_N^{\circ}=t_{\max}$.
	\end{algorithmic}
\end{algorithm}
Essentially, Algorithm 1 gives the optimal solution to Problem 3, denoted as $\bar{t}_N^{\circ}$, which is also a feasible solution to Problem 2 as per (\ref{lower}). However, due to the difference between $L(N,\bar{t})$ and $\lim_{K\rightarrow\infty}\frac{1}{K}\sum_{k=1}^K \text{tr} \mathbb{E}[P_k^e](\bar{t})$, $\bar{t}_N^{\circ}$ may not be the optimum to Problem 2. Its optimality is related to $N$ and $b(N,\bar{t})$. We will show in Theorem \ref{Thm7} the optimality condition and optimality gap.

\begin{theorem}
	\label{Thm7}
	Denote the optimal solution to Problem 2 as $\bar{t}^{\star}$, then $\bar{t}_N^{\circ}=\bar{t}^{\star}$ if 
	\begin{align}
	L(N,\bar{t}^{\circ}_N-1)+b(N,\bar{t}^{\circ}_N-1)<\underline{b}.\label{optimal}
	\end{align}
	Otherwise, the gap between $t^{\star}$ and $t^{\circ}_N$ is bounded as
	\begin{align}
	0 \leq t^{\circ}_N-t^{\star} \leq t^{\circ}_N -\min\{\bar{t}:L(N,\bar{t})+b(N,\bar{t})\geq \underline{b}\},\label{op_gap}
	\end{align}
	which indicates that $\lim_{N\rightarrow\infty}t^{\circ}_N = t^{\star}$.  \label{thm_op}
\end{theorem}
A proof is given in Appendix \ref{pf_thm7}.
Theorem \ref{thm_op} suggests that the solution we obtained from Problem 3 can be arbitrarily close to the optimum of Problem 2 by increasing the truncation horizon $N$. 

\section{Simulation}
We consider a linear time invariant system with system matrix
\begin{align}
A = & \left[\begin{array}{cc}
 0.95 & 0.85 \\
 0 & 0.99
\end{array}\right] \quad C = \left[\begin{array}{cc}
1 & 1
\end{array}\right]\nonumber\\
R = &0.01 \quad Q = \left[\begin{array}{cc}
0.0425 & 0.02\\
0.02 & 0.0425
\end{array}\right].\nonumber
\end{align}
It can be easily obtained that the spectral radius of $A$ is $0.99$ and the steady state error covariance $\bar{P}$ is computed as
\begin{align}
\bar{P} = \left[\begin{array}{cc}
0.0182 & -0.0139\\
-0.0139 & 0.0189
\end{array}\right].\nonumber
\end{align}
The packet reception probability is set as $\lambda= 0.3$ and the eavesdropping probability $\lambda_e = 0.3$. We test the performance of the scheduling method by varying the parameter $\underline{b}$, which denotes the user specified lower bound on eavesdropper's error covariance,  from $5$ to $100$. Other parameters in Algorithm 1 are taken as $N=300$, $\bar{t}_{\min}=1$ and $\bar{t}_{\max}=30$. A Monte Carlo simulation of length $10^6$ is conducted to obtain the values of $\text{tr}\mathbb{E}[P_k]$ and $\text{tr}\mathbb{E}[P_k^e]$. Fig. \ref{varying_b} (1) plots the corresponding transmission threshold obtained from Algorithm 1. Fig. \ref{varying_b} (2) compares the trace of the expected error covariance at the eavesdropper and that at the estimator, showing that the proposed scheduling strategy can efficiently suppress the growth of the expected error covariance at the legitimate estimator when the required privacy lower bound $\underline{b}$ increases. Meanwhile, we vary the truncation horizon $N$ from $35$ to $85$ to see its effects on the optimality. As shown in Fig. \ref{varying_N}, with the increase of $N$, the gap between the lower bound $L(N,\bar{t})$ and the true value of the expected error covariances gradually reduced, driving the threshold $\bar{t}_N^{\circ}$ to the optimum. 

\begin{figure}[h]
	\centering
	\includegraphics[width=0.8\textwidth]{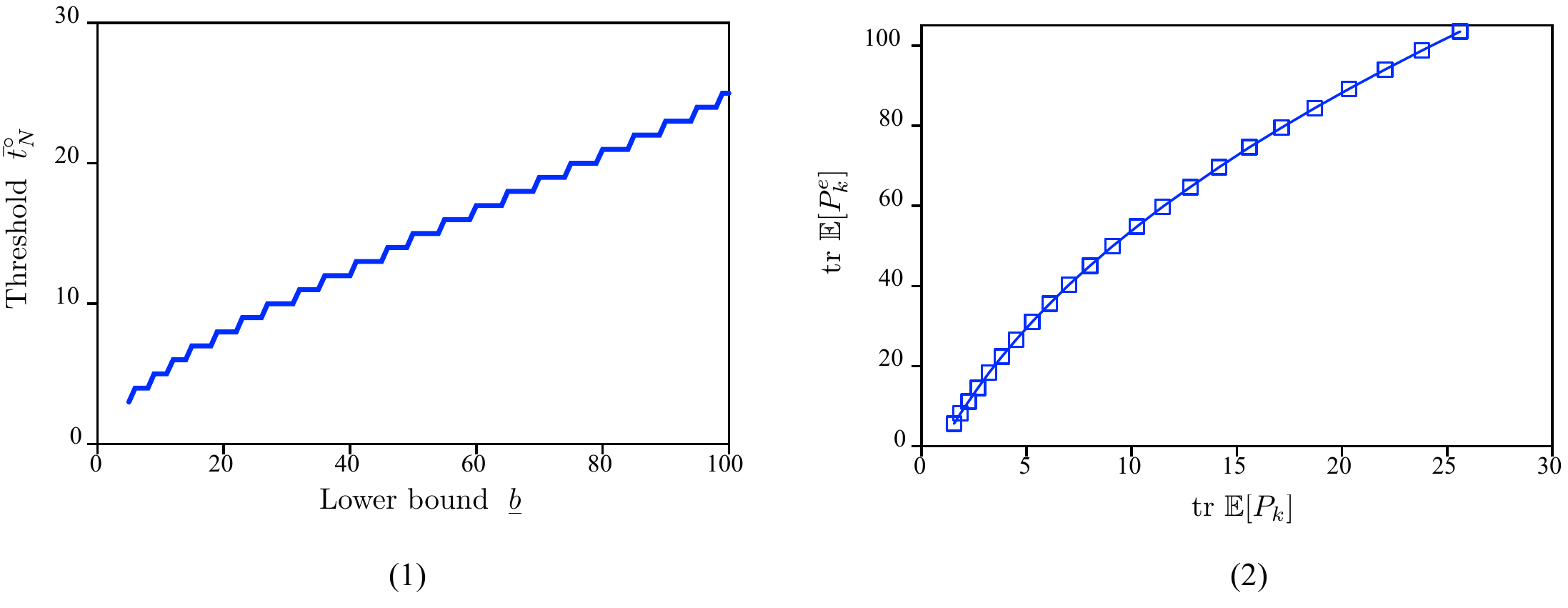}
	\caption{(1) Threshold $\bar{t}_N^{\circ}$ derived from Algorithm 1; (2) Trace of expected error covariance at estimator vs trace of expected error covariance at eavesdropper.} \label{varying_b}
\end{figure}

\begin{figure}[h]
	\centering
	\includegraphics[width=0.8\textwidth]{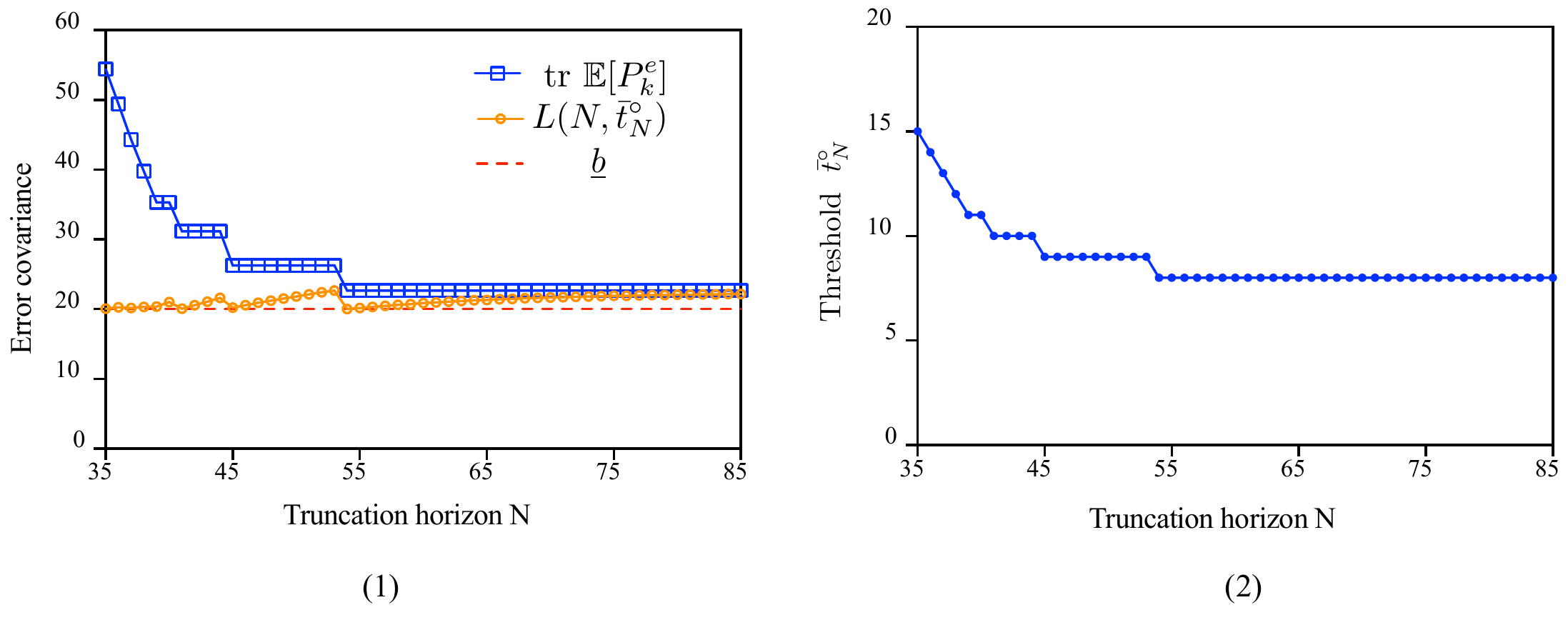}
	\caption{(1) Expected error covariance at the eavesdropper: lower bound $L(N,\bar{t})$ vs $\text{tr} \mathbb{E}[P_k^e]$ obtained from simulation  (2) Threshold $\bar{t}_N^{\circ}$ derived from Algorithm 1.} \label{varying_N}
\end{figure}

\section{Conclusion}
In this paper, we propose a threshold-type event-triggered transmission scheduling strategy for stable and marginally stable systems with eavesdroppers. An optimization problem is designed by minimizing the average of the expected error covariance at the authorized estimator and formulating a hard constraint to ensure the error covariance at the eavesdropper above a certain level. To solve this problem, first, we derive a recurrence relation to describe the stationary distribution of the expected error covariance at the eavesdropper; second, we show that the error covariance converges with the truncation horizon and monotonically increases with the growth of the transmission threshold. These two properties enable us to solve the optimization problem with a simple bisection method. Numerical simulations are conducted to verify the theoretic results and illustrate its potential in secure remote estimation. 



\appendix
\section{Proof to Proposition \ref{prop_mono}}
\label{pf_prop1}
\begin{proof}
	Denote the stationary distribution corresponding to the threshold $\bar{t}_1$ as $\pi_i(\bar{t}_1)$. 
	From (\ref{dis_pi}), it can be checked that $\pi_i(\bar{t}_1)>\pi_i(\bar{t}_1+1)$ if $0\leq i\leq \bar{t}_1$ and $\pi_i(\bar{t}_1)<\pi_i(\bar{t}_1+1)$ if $i\geq\bar{t}_1+1$. Considering that $\text{tr}f^i(\bar{P})$ monotonically increases with $i$, we have
	\begin{align}
	J(\bar{t}_1+1)-J(\bar{t}_1) = & \sum_{i=\bar{t}_1+1}^{\infty} \left(\pi_i(\bar{t}_1+1)-\pi_i(\bar{t}_1)\right)\text{tr}f^i(\bar{P}) - \sum_{i=0}^{\bar{t}_1} \left(\pi_i(\bar{t}_1)-\pi_i(\bar{t}_1+1)\right)\text{tr}f^i(\bar{P}) \nonumber\\
	\geq & \sum_{i=\bar{t}_1+1}^{\infty} \left(\pi_i(\bar{t}_1+1)-\pi_i(\bar{t}_1)\right)\text{tr}f^{\bar{t}_1+1}(\bar{P}) -  \sum_{i=0}^{\bar{t}_1} \left(\pi_i(\bar{t}_1)-\pi_i(\bar{t}_1+1)\right)\text{tr}f^{\bar{t}_1}(\bar{P}) \nonumber\\
	= & \left[\left(1-\sum_{i=0}^{\bar{t}_1}\pi_i(\bar{t}_1+1)\right)-\left(1-\sum_{i=0}^{\bar{t}_1}\pi_i(\bar{t}_1)\right) \right] \text{tr} f^{\bar{t}_1+1}(\bar{P})- \sum_{i=0}^{\bar{t}_1} \left(\pi_i(\bar{t}_1)-\pi_i(\bar{t}_1+1)\right)\text{tr}f^{\bar{t}_1}(\bar{P}) \nonumber\\
	= & \sum_{i=0}^{\bar{t}_1} \left(\pi_i(\bar{t}_1)-\pi_i(\bar{t}_1+1)\right) \left(\text{tr}f^{\bar{t}_1+1}(\bar{P}) -\text{tr}f^{\bar{t}_1}(\bar{P})  \right)>0,\nonumber
	\end{align}
	which completes the proof.  $\hfill\square$
	\end{proof}
\section{Proof to Lemma \ref{lemma2}}
\label{pfLemma2}
\begin{proof}
 By closely examining Eq. (\ref{mc1}), the following patterns can be derived,
\begin{equation}
\begin{aligned}
\phi_{0,0} = & \lambda\lambda_e\sum_{i=\bar{t}}^\infty \sum_{j\in\Omega_i} \phi_{i,j} \\
\phi_{0,j+1} = & \lambda(1-\lambda_e)\sum_{i\geq\bar{t},~i\in\Gamma_j}  \phi_{i,j} \\
\phi_{i+1,j+1} =  & \phi_{i,j},\quad 0 \leq i \leq \bar{t}-1,~j\in\Omega_j \\
\phi_{i+1,j+1} =  & \alpha \phi_{i,j},\quad i \geq \bar{t},~j\in\Omega_j 
\end{aligned} \label{pattern} 	
\end{equation}
which combining together with Eq.(\ref{mc2}) gives
\begin{equation}
\begin{aligned}
	\phi_{0,0} & = \lambda \lambda_e \sum_{i=\bar{t}}^\infty \sum_{j\in\Omega_i} \phi_{i,j} = \lambda\lambda_e \sum_{i=\bar{t}}^\infty \pi_i = \lambda\lambda_e(1-\bar{t}\frac{\lambda}{\lambda\bar{t}+1})\\
	& =  \frac{\lambda\lambda_e}{\lambda\bar{t}+1} \\
	\text{For}~0 \leq j< \bar{t}: \quad \phi_{0,j+1} & = \lambda(1-\lambda_e)\sum_{i\geq\bar{t},~i\in\Gamma_j}  \phi_{i,j} = 
	 \lambda(1-\lambda_e) \sum_{i\in[\bar{t}+j+1,\infty)}\phi_{i,j} =  \lambda(1-\lambda_e) (1-\lambda)(1-\lambda_e)\sum_{j\in[\bar{t}+j,\infty)}\phi_{i,j-1} \\
	 & = (1-\lambda)(1-\lambda_e) \underbrace{\left[\lambda(1-\lambda_e)\sum_{j\in[\bar{t}+j,\infty)}\phi_{i,j-1}\right]}_{=\phi_{0,j-1}}\\
	& = \alpha\phi_{0,j-1} \\
		\text{For}~ j\geq \bar{t}: \quad \phi_{0,j+1} & =  \lambda(1-\lambda_e)\sum_{i\geq\bar{t},~i\in\Gamma_j}  \phi_{i,j} = \beta \left(\phi_{\bar{t},j}+\sum_{i\in[\bar{t}+1,j]\cup[\bar{t}+j+1,\infty)}\phi_{i,j}\right) \\
		& =  \beta \phi_{0,j-\bar{t}} + \beta \alpha \sum_{i\in[\bar{t},j-1]\cup[\bar{t}+j,\infty)}\phi_{i,j-1}\\
		& = \beta \phi_{0,j-\bar{t}} +  \alpha \left(\beta \sum_{i\in[\bar{t},j-1]\cup[\bar{t}+j,\infty)}\phi_{i,j-1}\right)\\
		& = \beta \phi_{0,j-\bar{t}} +  \alpha \phi_{0,j}.
\end{aligned} \nonumber
\end{equation}
 The recurrence relation in (\ref{lemma23}) is thus obtained.  $\hfill\square$
\end{proof}

\section{Proof to Proposition \ref{prop3}}
\label{pfprop3}
\begin{proof}
	Note that the recursion in (\ref{lemma23}) can be viewed as a linear time invariant system with constant time delay $\bar{t}+1$. Applying Theorem 1 in \cite{diblik2016exponential}, it is straightforward to obtain
	\begin{align}
	0 \leq \phi_{0,j} \leq \gamma^{j-\bar{t}}\max_{l\in[0,\bar{t}]}\phi_{0,l}=\gamma^{j-\bar{t}}\phi_{0,0},\nonumber
	\end{align}
and $0\leq \lim_{j\rightarrow \infty} \phi_{0,j}\leq \lim_{j\rightarrow \infty} \gamma^{j-\bar{t}}\phi_{0,0} = 0$. In addition, summing up $\phi_{0,j}$ from $j=0$ to $j=n$ according to (\ref{lemma23}) gives
	\begin{align}
	\sum_{j=0}^{n+1} \phi_{0,j} = \left\{\begin{array}{ll}
	\frac{\lambda\lambda_e}{\lambda\bar{t}+1}+\alpha \sum_{j=0}^n \phi_{0,j} & n\leq \bar{t}-1\\
	\frac{\lambda\lambda_e}{\lambda\bar{t}+1}+\alpha \sum_{j=0}^n \phi_{0,j} + \beta \sum_{j=0}^{n-\bar{t}}\phi_{0,j}& n\geq \bar{t} .
	\end{array}\right. \label{sum_recursion}
	\end{align}
For $n\leq \bar{t}_1$ and $\bar{t}_1\leq \bar{t}_2$, it is straightforward to see that $\sum_{j=0}^n \phi_{0,j}(\bar{t}_1) \geq \sum_{j=0}^n \phi_{0,j}(\bar{t}_2) $. For any $n>\bar{t}_1$, if $\sum_{j=0}^n \phi_{0,j}(\bar{t}_1) \geq \sum_{j=0}^n \phi_{0,j}(\bar{t}_2) $ and $\sum_{j=0}^{n-\bar{t}} \phi_{0,j}(\bar{t}_1) \geq \sum_{j=0}^{n-\bar{t}} \phi_{0,j}(\bar{t}_2) $, then $\sum_{j=0}^{n+1} \phi_{0,j}(\bar{t}_1) \geq \sum_{j=0}^{n+1} \phi_{0,j}(\bar{t}_2) $ according to (\ref{sum_recursion}). Thus, the inequality in (\ref{phi_mono}) can be proved using mathematical induction.    	$\hfill\square$
\end{proof}

\section{Proof to Theorem \ref{thm5}}
\label{pfThm5}

\begin{proof}
By iteratively applying the relations in (\ref{pattern}), the recurrence relation of $\omega_j$ can be established in a similar manner as that for $\phi_{0,j}$ . In particular, we can show that
	\begin{align}
	\omega_0 = & \sum_{i\in\Gamma_0}\phi_{i,0}=\lambda_e \sum_{i\in[\bar{t},\infty),~j\in\Omega_i}\phi_{i,j}=\lambda_e\sum_{i=\bar{t}}^{\infty}\pi_i=\lambda_e(1-\bar{t}\pi_0)\nonumber\\
	= & \frac{\lambda_e}{\lambda\bar{t}+1} \nonumber\\
	\text{For}~{j\leq\bar{t}}:\quad \omega_j= & \sum_{i\in\Gamma_j}\phi_{i,j}= \underbrace{\phi_{j,j}}_{=\phi_{0,0}} + \underbrace{\sum_{i=0}^{j-1}\phi_{i,j}}_{=\alpha \sum_{i=0}^{j-1}\phi_{i,j-1}} + \underbrace{\sum_{i=\bar{t}+j+1}^{\infty}\phi_{i,j}}_{=\alpha \sum_{i=\bar{t}+j}^{\infty}\phi_{i,j-1}} \nonumber\\
	= & \phi_{0,0} + \alpha (\sum_{\footnotesize \underbrace{i\in[0,j-1]\cup[\bar{t}+j,\infty)}_{i\in\Gamma_{j-1}}}\phi_{i,j-1})\nonumber \\
	= &  \phi_{0,0} + \alpha \omega_{j-1}\nonumber  \\
	\text{For}~{j>\bar{t}}:\quad \omega_j= &  \sum_{l=0}^{\bar{t}-1} \phi_{l,j} + \phi_{\bar{t},j}+\sum_{l=\bar{t}+1}^{j} \phi_{l,j} + \sum_{l=\bar{t}+1+j}^{\infty} \phi_{l,j}\nonumber\\
	= & \alpha\sum_{l=0}^{\bar{t}-1}\phi_{l,j-1} + \beta \sum_{l=\max(0,j-2\bar{t})}^{j-1-\bar{t}}\phi_{0,l} +  \phi_{\bar{t},j} + \alpha \sum_{l=\bar{t}}^{j-1} \phi_{l,j-1}+ \alpha \sum_{l=\bar{t}+j}^{\infty} \phi_{l,j-1} \nonumber\\
	= & \alpha (\sum_{l=0}^{\bar{t}-1}\phi_{l,j-1} +  \sum_{l=\bar{t}}^{j-1} \phi_{l,j-1}+\sum_{l=\bar{t}+j}^{\infty} \phi_{l,j-1} ) + \beta \sum_{l=\max(0,j-2\bar{t})}^{j-1-\bar{t}}\phi_{0,l} +  \phi_{0,j-\bar{t}}\nonumber\\
	= & \alpha \omega_{j-1} + \beta \sum_{l=\max(0,j-2\bar{t})}^{j-1-\bar{t}}\phi_{0,l} +  \phi_{0,j-\bar{t}}.\label{omega_eq}
	\end{align}
Define
\vspace{-2mm}
\begin{align}
\Theta_j =\left\{\begin{array}{ll}\phi_{0,0} & 1\leq j\leq \bar{t}\\
\beta \sum_{l=\max(0,j-2\bar{t})}^{j-1-\bar{t}}\phi_{0,l} +  \phi_{0,j-\bar{t}}  &  j>\bar{t}
\end{array}
\right. .\label{phi}
\end{align}
Eq. (\ref{phi}) becomes
\begin{align}
\omega_j = \alpha \omega_{j-1} + \Theta_j.\nonumber
\end{align}
For $j>\bar{t}$, we have
\vspace{-4mm}
	\begin{align}
\Theta_{j+1}= & \beta \sum_{l=\max(0,j+1-2\bar{t})}^{j-\bar{t}}\phi_{0,l} +  \phi_{0,j+1-\bar{t}} \nonumber\\
= & \left\{\begin{array}{ll}
\beta \sum_{l=0}^{j-1-\bar{t}}\phi_{0,l} + \underbrace{\beta \phi_{0,j-\bar{t}} + \alpha \phi_{0,j-\bar{t}}}_{=\phi_{0,j-\bar{t}} - (1-\alpha -\beta ) \phi_{0,j-\bar{t}}=-\lambda_e\phi_{-,j-\bar{t}}} & j\leq 2\bar{t}-1\\
\beta \sum_{l=j+1-2\bar{t}}^{j-1-\bar{t}}\phi_{0,l} + \beta \phi_{0,j-\bar{t}} + \alpha \phi_{0,j-\bar{t}} +\beta \phi_{0,j-2\bar{t}} &  j>2\bar{t}-1
\end{array} \right. \quad\text{(according to (\ref{lemma23}))} \nonumber \\
= & \Theta_j-\lambda_e\phi_{0,j-\bar{t}}.\label{pfthm_2}
\end{align}
Using mathematical induction, it can be derived that 
 \begin{align}\Theta_j = \phi_{0,0}-\lambda_e\sum_{l=0}^{j-1-\bar{t}}\phi_{0,l}.  \label{pfthm_3}
	\end{align} 
Since $\phi_{0,l}\geq0$ for each $l$, the proof for (\ref{recursion}) is completed.  
 
\textbf{(\RNum{1})} Eq. (\ref{pfthm_2}) shows that $\Theta_{j+1}\leq \Theta_{j}$ for any $j>0$, based on which (\ref{thm_mono}) can be easily proved. 
	
\textbf{(\RNum{2})} The boundedness and exponential convergence of $\omega_j$ can be obtained from Eq. (\ref{phi}) and (\ref{omega_eq}). For $j>3\bar{t}$, from (\ref{phi_bounded}), we have
	\begin{align}
	\Theta_j \leq  (1+\beta\bar{t})\phi_{0,0}\gamma^{j-3\bar{t}}.\nonumber
	\end{align}
By iteratively applying Eq. (\ref{omega_eq}), we have
	\begin{align}
	\omega_j \leq \alpha^{j-N} \omega_N + \sum_{l=N+1}^{j}\alpha^{j-l}\Theta_{l}\leq \alpha^{j-N} \omega_N + (j-N)(1+\beta\bar{t})\phi_{0,0}\gamma^{j-3\bar{t}}
	\end{align}
	since $0<\alpha< \gamma<1$.
	
	\textbf{(\RNum{3})} Note that if the inequality in (\ref{omega_mono}) holds for $\bar{t}_2=\bar{t}_1+1$, then it holds for all $\bar{t}_2>\bar{t}_1$. From (\ref{lemma23}) and (\ref{phi_mono}), we have
	 \begin{align}
	\omega_0(\bar{t}_1)+\sum_{j=1}^n\Theta_j(\bar{t}_1)\geq \omega_0(\bar{t}_2)+\sum_{j=1}^n  \Theta_j(\bar{t}_2),\label{sum_mono}
	\end{align}  
	for any $n\geq 1$. Again applying mathematical induction with respect to (\ref{omega_eq}), we have
			\begin{align}
	\sum_{j=0}^n \omega_j(\bar{t}_1) \geq \sum_{j=0}^n \omega_j(\bar{t}_2) \nonumber
	\end{align}
	which is equivalent to (\ref{omega_mono}).
	
	\textbf{(\RNum{4})} Combining $\sum_{j=0}^{\infty}\omega_j(\bar{t}_1) = \sum_{j=0}^{\infty}\omega_j(\bar{t}_2)=1$ and (\ref{omega_mono}), we have
	\begin{align}
	\sum_{j=1}^{\infty} \omega_j(\bar{t}_2)-\sum_{j=1}^{\infty} \omega_j(\bar{t}_1) \geq - (\omega_0(\bar{t}_2)-\omega_0(\bar{t}_1)).\nonumber
	\end{align}  
	Since $0\leq \text{tr} f^0(\bar{P})\leq \text{tr} f^1(\bar{P})$, we have
	\begin{align}
	\text{tr} f^1(\bar{P})\left(\sum_{j=1}^{\infty} \omega_j(\bar{t}_2)-\sum_{j=1}^{\infty} \omega_j(\bar{t}_1) \right) \geq  - \text{tr} f^0(\bar{P})(\omega_0(\bar{t}_2)-\omega_0(\bar{t}_1)),\nonumber
	\end{align}	
	thus, 
		\begin{align}
	\text{tr} f^1(\bar{P})\left(\sum_{j=2}^{\infty} \omega_j(\bar{t}_2)-\sum_{j=2}^{\infty} \omega_j(\bar{t}_1) \right) \geq  - \sum_{j=0}^1\text{tr} f^j(\bar{P})(\omega_j(\bar{t}_2)-\omega_j(\bar{t}_1)).\nonumber
	\end{align}	
	Repeating the above steps, we have
	\begin{align}
	\text{tr} f^2(\bar{P})\left(\sum_{j=2}^{\infty} \omega_j(\bar{t}_2)-\sum_{j=2}^{\infty} \omega_j(\bar{t}_1) \right) \geq & 	\text{tr} f^1(\bar{P})\left(\sum_{j=2}^{\infty} \omega_j(\bar{t}_2)-\sum_{j=2}^{\infty} \omega_j(\bar{t}_1) \right) \geq  - \sum_{j=0}^1\text{tr} f^j(\bar{P})(\omega_j(\bar{t}_2)-\omega_j(\bar{t}_1)) \nonumber\\ \text{tr} f^3(\bar{P})\left(\sum_{j=3}^{\infty} \omega_j(\bar{t}_2)-\sum_{j=3}^{\infty} \omega_j(\bar{t}_1) \right) \geq & - \sum_{j=0}^2\text{tr} f^j(\bar{P})(\omega_j(\bar{t}_2)-\omega_j(\bar{t}_1)) \nonumber\\
	\vdots & \nonumber
	\end{align}
	which leads to $\sum_{j=0}^{\infty} \Big(\omega_j(\bar{t}_2)\text{tr}f^j(\bar{P})-\omega_j(\bar{t}_1)\text{tr}f^j(\bar{P})\Big)\geq 0$, completing the proof to (\ref{mono_average}). $\hfill\square$
\end{proof}

\section{Proof to Theorem \ref{thm6}}
\label{pfThm6}
\begin{proof}
	
	(\textbf{\RNum{1}}) Since $1-\sum_{j=0}^N\omega_j=\sum_{j=N+1}^{\infty}\omega_j$ and $\text{tr}f^{N+1}(\bar{P})\leq \text{tr}f^{N+2}(\bar{P})\leq \text{tr}f^{N+3}(\bar{P}) \leq \dots$, we have
	\begin{align}
	(1-\sum_{j=0}^N\omega_j)\text{tr}f^{N+1}(\bar{P}) \leq \sum_{j=N+1}^{\infty} \omega_j \text{tr} f^j(\bar{P}),\nonumber
	\end{align}
	which proves the inequality in  (\ref{lower}).
	
	(\textbf{\RNum{2}}) Given that $\rho(A)\leq 1$, according to the monotonicity in (\ref{p_mono}),  we can derive an upper bound on $\text{tr}f^i(\bar{P})$ for any given $i$ as
	\begin{align}
	\text{tr}f^i(\bar{P})\leq & \lim_{k\rightarrow\infty} \text{tr}f^k(\bar{P})\nonumber\\
	= &\lim_{k\rightarrow\infty} \text{tr}A^k\bar{P}(A^\top)^k+\lim_{k\rightarrow\infty} \text{tr}\left(\sum_{m=0}^{k-1}A^mQ(A^\top)^m\right) \nonumber\\ 
	\leq & \lim_{k\rightarrow\infty}n_s \rho\left(A^k\bar{P}(A^\top)^k\right) + \lim_{k\rightarrow\infty}n_s\rho\left(\sum_{m=0}^{k-1}A^mQ(A^\top)^m\right) \nonumber\\
	\leq & \lim_{k\rightarrow\infty}n_s\rho(\bar{P})  \left[\rho(A)^{2k} + n_s\rho(Q)\sum_{m=0}^{k-1}\rho(A)^{2m}\right].\nonumber
	\end{align}
Here $n_s$ denotes the dimension of matrix $A$. If $\rho(A)<1$, we have
\begin{align}
\lim_{k\rightarrow\infty}n_s\rho(\bar{P})  \left[\rho(A)^{2k} + n_s\rho(Q)\sum_{m=0}^{k-1}\rho(A)^{2m}\right] = \frac{n_s\rho(Q)}{1-\rho(A)^2}.\nonumber
\end{align}
If $\rho(A)=1$, we have
\begin{align}
n_s\rho(\bar{P})  \left[\rho(A)^{2k} + n_s\rho(Q)\sum_{m=0}^{k-1}\rho(A)^{2m}\right]=n_s\rho(\bar{P})+n_s\rho(Q)k.\nonumber
\end{align}
Therefore, for $\rho(A)<1$, it can be obtained that 	
	\begin{align}
	\sum_{j=N+1}^{\infty} \omega_j \text{tr} (f^j(\bar{P})) \leq & \left(\sum_{j=N+1}^{\infty} \omega_j \right)	\frac{n_s\rho(Q)}{1-\rho(A)^2}
	\leq  \left(\frac{\alpha}{1-\alpha}\omega_N +\gamma^{N-3\bar{t}}(1+\beta\bar{t})\phi_{0,0}\frac{\gamma}{(1-\gamma^2)} \right) \frac{n_s\rho(Q)}{1-\rho(A)^2}.\nonumber
	\end{align}
	Similarly, for $\rho(A)=1$, it can be proved that
	\begin{align}
	\sum_{j=N+1}^{\infty} \omega_j \text{tr} (f^j(\bar{P})) \leq  & \sum_{j=N+1}^{\infty}  \left(\alpha^{j-N} \omega_N + (j-N)(1+\beta \bar{t})\phi_{0,0}\gamma^{j-3\bar{t}}\right)\left(	n_s\rho(\bar{P})+n_s\rho(Q)j\right) \nonumber\\
	\leq & \left(\frac{\alpha}{1-\alpha}(n_s\rho(\bar{P})+n_s\rho(Q)N)+\frac{\alpha}{(1-\alpha)^2}n_s\rho(Q)\right)\omega_N \nonumber\\
	&+ (1+\beta\bar{t})\phi_{0,0}n_s\gamma^{N-3\bar{t}}\left((\rho(\bar{P})+N\rho(Q))\frac{\gamma}{(1-\gamma)^2}+\rho(Q)\frac{\gamma}{(1-\gamma)^3}\right).\nonumber
	\end{align}
	Thus, define $b(N,\bar{t})$ as
	\begin{align}
	b(N,\bar{t})=\left\{
	\begin{array}{ll}
\left(\frac{\alpha}{1-\alpha}\omega_N +\gamma^{N-3\bar{t}}\phi_{0,0}(1+\beta\bar{t})\frac{\gamma}{(1-\gamma^2)} \right) \frac{n_s\rho(Q)}{1-\rho(A)^2}-(1-\sum_{j=0}^N\omega_j) \text{tr}f^{N+1}(\bar{P}) &  \rho(A)<1\\
	\left(\frac{\alpha}{1-\alpha}(n_s\rho(\bar{P})+n_s\rho(Q)N)+\frac{\alpha}{(1-\alpha)^2}n_s\rho(Q)\right)\omega_N + (1+\beta\bar{t})\phi_{0,0}n_s\gamma^{N-3\bar{t}} & \rho(A) = 1\\
		\times \left((\rho(\bar{P})+N\rho(Q))\frac{\gamma}{(1-\gamma)^2}+\rho(Q)\frac{\gamma}{(1-\gamma)^3}\right) -(1-\sum_{j=0}^N\omega_j) \text{tr}f^{N+1}(\bar{P}) &
	\end{array}	
	\right. . \nonumber
	\end{align}
	Since the terms $\omega_N$, $1-\sum_{j=0}^N\omega_j$ and $\gamma^{N}$ converge to $0$, we have $\lim_{N\rightarrow\infty}b(N,\bar{t}) = 0$ for any $\bar{t}$, which completes the proof. 
	
	\textbf{(\RNum{3})} The inequality in (\ref{lower_bound}) can be proved  by similar arguments as in the proof to Item (\RNum{4}) of Theorem \ref{thm5}. $\hfill\square$
\end{proof}

\nopagebreak[4]
\section{Proof to Theorem \ref{Thm7} }
\label{pf_thm7}
\begin{proof}
Problem 2's optimal solution $\bar{t}^{\star}$ should be the smallest integer making $\sum_{j=0}^{\infty}\omega_j \text{tr} (f^j(\bar{P}) \geq \underline{b}$ hold. Then, we have
\begin{align}
\sum_{j=0}^{\infty}\omega_j (\bar{t}^{\star}-1)\text{tr} f^j(\bar{P}) <& \underline{b},\label{op_c1}\\
\sum_{j=0}^{\infty}\omega_j (\bar{t}^{\star})\text{tr} f^j(\bar{P}) \geq & \underline{b}.\label{op_c2}
\end{align}	
Since the function $\sum_{j=0}^{\infty}\omega_j (\bar{t})\text{tr} (f^j(\bar{P}))$ is monotone with $\bar{t}$, for any $\bar{t}$ which satisfies (\ref{op_c1}) and (\ref{op_c2}), $\bar{t}$ should be the optimal solution. As it is proved that $L(N,\bar{t})+b(N,\bar{t})$ is an upper bound of $\sum_{j=0}^{\infty}\omega_j (\bar{t})\text{tr} (f^j(\bar{P}))$, any $\bar{t}_N^{\circ}$ satisfying 
(\ref{optimal}) will give
\begin{align}
\sum_{j=0}^{\infty}\omega_j (\bar{t}_N^{\circ}-1)\text{tr} f^j(\bar{P}) < \underline{b}.\nonumber
\end{align}
In addition, $\bar{t}_N^{\circ}$ is a feasible solution to Problem 3, ensuring that 
\begin{align}
\sum_{j=0}^{\infty}\omega_j (\bar{t}_N^{*})\text{tr} f^j(\bar{P}) \geq & \underline{b}.\nonumber
\end{align}
Therefore, given that $\bar{t}_N^{\circ}$ satisfies (\ref{optimal}), we have $\bar{t}_N^{\circ}=\bar{t}^{\star}$. 

Given that $\bar{t}^{\star}$ is the optimal solution to Problem 2, $\bar{t}^{\star}$ must be feasible for the constraint and we have
\begin{align}
L(N,\bar{t}^{\star}) + b(N,\bar{t}^{\star})\geq\sum_{j=0}^{\infty}\omega_j (\bar{t}^{\star})\text{tr} (f^j(\bar{P})\geq \underline{b}.\nonumber
\end{align}
Thus, $\bar{t}^{\star}\geq \min\{\bar{t}:L(N,\bar{t})+b(N,\bar{t})\geq \underline{b}\}$, completing the proof to (\ref{op_gap}). When $N\rightarrow\infty$, we have $\lim_{N\rightarrow\infty}b(N,\bar{t})=0$ for any bounded $\bar{t}$, thus, $\lim_{N\rightarrow\infty} t^{\circ}_N -\min\{\bar{t}:L(N,\bar{t})+b(N,\bar{t})\geq \underline{b}\}=0$. Then, it can be proved by the sandwich theorem that $\lim_{N\rightarrow\infty} t^{\circ}_N=\bar{t}^{\star}$. $\hfill\square$
\end{proof}

\bibliography{sample}

\end{document}